\documentclass[10pt, conference, a4paper]{IEEEtran}

\usepackage{graphicx,colordvi,psfrag}
\usepackage{amsmath,amssymb}
\usepackage{epstopdf}
\usepackage[caption=false]{subfig}
\usepackage{epsfig,cite}
\usepackage{calc,pstricks, pgf, xcolor}
\usepackage{nicefrac}
\usepackage{enumerate}
\usepackage{dsfont}
\usepackage{bm}

\newtheorem{theorem}{Theorem}
\newtheorem{corollary}{Corollary}
\newtheorem{remark}{Remark}

\newtheorem{proposition}{Proposition}

\newenvironment{proof}[1][Proof]{\noindent\textbf{#1.} }{\ \rule{0.5em}{0.5em}}

\newcommand{\bY}{\mathbf{Y}}

\newcommand{\bX}{\mathbf{X}}

\newcommand{\bP}{\mathbf{P}}

\newcommand{\Ber}{\mathrm{Bernoulli}}
\newcommand{\Bin}{\mathrm{Binomial}}

\newcommand{\m}{\mathcal}

\def\bpi{\bm{\pi}}

\begin{document}

\IEEEoverridecommandlockouts

\title{Novel Lower Bounds on the Entropy Rate of Binary Hidden Markov Processes}
\author{\IEEEauthorblockN{Or Ordentlich}
\IEEEauthorblockA{MIT\\
ordent@mit.edu
}

\IEEEcompsocitemizethanks{
\IEEEcompsocthanksitem
The work of O. Ordentlich was supported by the MIT - Technion Postdoctoral Fellowship.
}}

\parskip 3pt

\maketitle

\begin{abstract}
Recently, Samorodnitsky proved a strengthened version of Mrs. Gerber's Lemma, where the output entropy of a binary symmetric channel is bounded in terms of the average entropy of the input projected on a random subset of coordinates. Here, this result is applied for deriving novel lower bounds on the entropy rate of binary hidden Markov processes. For symmetric underlying Markov processes, our bound improves upon the best known bound in the very noisy regime. The nonsymmetric case is also considered, and explicit bounds are derived for Markov processes that satisfy the $(1,\infty)$-RLL constraint.
\end{abstract}

\section{Introduction}\label{sec:intro}

Let $\{X_n\}$, $n=1,2,\ldots$, be a symmetric stationary binary Markov process with transition probability $0<q<\tfrac{1}{2}$, such that $X_1\sim\Ber(\tfrac{1}{2})$ and for any $n>1$
\begin{align}
X_n=X_{n-1}\oplus W_n,\nonumber
\end{align}
where $\{W_n\}$, $n=2,3,\ldots$, is a sequence of i.i.d. $\Ber(q)$ random variables, statistically independent of $X_1$. We consider the hidden Markov process $\{Y_n\}$, $n=1,2,\ldots$, obtained at the output of a binary symmetric channel (BSC) with crossover probability $0<\alpha<\tfrac{1}{2}$, whose input is the process $\{X_n\}$. Namely,
\begin{align}
Y_n=X_n\oplus Z_n,\nonumber
\end{align}
where $\{Z_n\}$, $n=1,2,\ldots$, is a sequence of i.i.d. $\Ber(\alpha)$ random variables, statistically independent of $\{X_n\}$. The task of finding an explicit form for the entropy rate
\begin{align}
\bar{H}(Y)\triangleq\lim_{n\to\infty}\frac{H(Y_1,\ldots,Y_n)}{n}\nonumber
\end{align}
of the process $\{Y_n\}$ is a long-standing open problem, and the main contribution of this paper is in providing novel lower bounds for this quantity.

A simple lower bound on $\bar{H}(Y)$ can be obtained by invoking Mrs. Gerber's Lemma (MGL)~\cite{wz73}, which states that if $\{X_n\}$ is the input to a BSC with crossover probability $\alpha$, and $\{Y_n\}$ is the output, then
\begin{align}
H(Y_1,\ldots,Y_n)\geq n h\left(\alpha*h^{-1}\left(\frac{H(X_1,\ldots,X_n)}{n}\right)\right),\label{eq:MGL}
\end{align}
where $h(p)\triangleq -p\log(p)-(1-p)\log(1-p)$ is the binary entropy function, $h^{-1}(\cdot)$ is its inverse restricted to the interval $[0,\tfrac{1}{2}]$, and $a*b\triangleq a(1-b)+b(1-a)$. Here, as well as throughout the rest of the paper, logarithms are taken to base $2$. Since the entropy rate of the symmetric Markov process $\{X_n\}$ is $\bar{H}(X)=h(q)$, for symmetric hidden Markov processes the bound~\eqref{eq:MGL} takes the simple form
\begin{align}
\bar{H}(Y)\geq h\left(\alpha*q\right).\label{eq:mglmarkov}
\end{align}
Unfortunately, this bound is quite loose for many regimes of the process parameters $\alpha$ and $q$.

Recently, Samorodnitsky~\cite{Samorodnitsky15} proved a strengthened version of MGL, where the normalized input entropy $H(X_1,\ldots,X_n)/n$ in the right hand side of~\eqref{eq:MGL} is replaced by the average normalized entropy of the random vector $(X_1,\ldots,X_n)$ projected on a random subset of coordinates.
In this paper we apply the results of~\cite{Samorodnitsky15} to derive a novel lower bound on $\bar{H}(Y)$. Despite its simplicity, we show that this bound is stronger than the best known lower bounds for the very noisy regime ($\alpha\to \tfrac{1}{2}$), and recovers the strongest bound for the fast transitions regime ($q\to\tfrac{1}{2}$). For finite values of $(\alpha,q)$ it is numerically demonstrated that the bound is reasonably close to the true value of $\bar{H}(Y)$, which can be estimated to an arbitrary precision by various known approximation algorithms.

We also derive a lower bound on $\bar{H}(Y)$ for the case where the process $\{X_n\}$ is a nonsymmetric binary Markov process. For the special case of Markov processes that satisfy the so-called $(1,\infty)$-RLL constraint, our bound is shown to be tight in the very noisy regime.

%The remainder of the paper is organized as follows. Section~\ref{sec:prel} gives some necessary background for the derivation of our bound, brought in Section~\ref{sec:main}. The behavior of the bound for the various asymptotic regimes of the parameters $\alpha$ and $q$ is analyzed in Section~\ref{sec:asymptotics}, where numerical comparisons to existing approximation algorithms for $\bar{H}(Y)$ are also performed. Extensions of our bound to nonsymmetric binary Markov processes are derived in Section~\ref{sec:nonsym}, and for the special case of binary Markov processes with no consecutive $1$s, explicit expressions are developed.

\section{Preliminaries}\label{sec:prel}

Let $\bX=(X_1,\ldots,X_n)$ be a binary $n$-dimensional random vector, $[n]\triangleq\{1,\ldots,n\}$, and $S\subseteq[n]$ some subset of coordinates. The projection of $\bX$ onto $S$ is defined as
\begin{align}
\bX_S\triangleq \{X_i \ : i\in S\}.\nonumber
\end{align}
As before, we assume that $\bY$ is the output of a BSC with crossover probability $\alpha$, whose input is the vector $\bX$. Samorodnitsky has proved the following result.

\begin{theorem}[{\cite[Theorem 1.11]{Samorodnitsky15}}]
Let $\lambda=(1-2\alpha)^2$ and let $B$ be a random subset uniformly distributed over all subsets of $[n]$ with cardinality $\lceil \lambda n\rceil$. Then
\begin{align}
H(\bY)\geq n h\left(\alpha*h^{-1}\left(\frac{H(\bX_B|B)}{\lambda n}\right)\right)-E,\label{eq:SMGL}
\end{align}
where $E=\m{O}\left(\sqrt{\frac{\log n}{n}}\right)\cdot(n-H(\bX))$.
\label{thm:smgl}
\end{theorem}

By Han's inequality~\cite{coverthomas}, the quantity $H(\bX_B|B)/\lambda n$ is monotonically nonincreasing in $\lambda$, and therefore, ignoring the error term $E$, it can be seen that the bound~\eqref{eq:SMGL} is stronger than~\eqref{eq:MGL}.

For our purposes, it will be convenient to replace $H(\bX_B|B)$ with $H(\bX_S|S)$, where $S$ is a random subset of $[n]$ generated by independently sampling each element $i$ with probability $\lambda$. It is easy to verify that for any distribution $P_{\bX}$ on $\{0,1\}^n$ holds
\begin{align}
\lim_{n\to\infty}\frac{H(\bX_B|B)-H(\bX_S|S)}{\lambda n}=0,\nonumber
\end{align}
and we can therefore indeed replace $B$ with $S$ in Theorem~\ref{thm:smgl}, perhaps with a different convergence rate for $E$. In fact, Polyanskiy and Wu~\cite{pw16} distilled from~\cite{Samorodnitsky15} the inequality
\begin{align}
I(U:\bY)\leq I(U:\bX_S|S),\label{eq:PW}
\end{align}
that holds for any random variable $U$ satisfying the Markov relation $U\to\bX\to\bY$. Using~\eqref{eq:PW}, the chain rule of entropy, and the convexity of the MGL function $\varphi(t)\triangleq h(\alpha*h^{-1}(t))$, it is a simple exercise to prove the following form of Theorem~\ref{thm:smgl}.
\begin{proposition}
Let $\lambda=(1-2\alpha)^2$ and let $S$ be a random subset of $[n]$ generated by independently sampling each element
$i$ with probability $\lambda$. Then
\begin{align}
H(\bY)\geq n h\left(\alpha*h^{-1}\left(\frac{H(\bX_S|S)}{\lambda n}\right)\right).\label{eq:SMGL2}
\end{align}
\label{prop:smgl2}
\end{proposition}

\section{Main Result}\label{sec:main}

In order to apply Proposition~\ref{prop:smgl2} for lower bounding $\bar{H}(Y)$, we need to evaluate the quantity $H(\bX_S|S)/\lambda n$ for symmetric Markov processes $\{X_n\}$. We will use the notation
\begin{align}
q^{*k}\triangleq \underbrace{q*q*\cdots*q}_{k \text{ times}},\nonumber
\end{align}
and note that
\begin{align}
q^{*k}=\Pr\left(W_1\oplus\cdots \oplus W_k =1\right)=\frac{1-(1-2q)^k}{2}.\label{eq:qstar}
\end{align}

\begin{proposition}
Let $0< \lambda< 1$, and let $S$ be a random subset of $[n]$ generated by independently sampling each element
$i$ with probability $\lambda$. Then
\begin{align}
\lim_{n\to\infty}\frac{H(\bX_S|S)}{\lambda n}=\mathbb{E}h\left(q^{*G}\right),\nonumber
\end{align}
where $G$ is a geometric random variable with parameter $\lambda$, i.e. $\Pr(G=g)=(1-\lambda)^{g-1}\lambda$ for $g=1,2,\ldots$.
\label{prop:MarkovBEC}
\end{proposition}

\begin{proof}
Let $G_i$, $i=1,2,\ldots$, be a sequence of i.i.d. geometric random variables with parameter $\lambda$. Define the autoregressive process
\begin{align}
A_k=\sum_{i=1}^k G_k, \ k=1,2,\ldots,\nonumber
\end{align}
and define the random variable $K$ as the largest $k$ for which $A_k\leq n$. Clearly, the subset $S$ and the subset $\{A_1,\ldots,A_K\}$ have the same distribution, and therefore
\begin{align}
&\frac{H(\bX_S|S)}{n} =\frac{1}{n}\mathbb{E}\sum_{i=1}^K H(X_{A_i}|X_{A_{i-1}},\ldots,X_{A_1})\nonumber\\
&=\frac{1}{n} \mathbb{E}\left(\sum_{i=1}^n H(X_{A_i}|X_{A_{i-1}},\ldots,X_{A_1})\mathds{1}(i\leq K)\right)\nonumber\\
&=\frac{1}{n} \mathbb{E}\left(\sum_{i=1}^n H(X_{A_i}|X_{A_{i-1}})\mathds{1}(i\leq K)\right)\label{eq:Markovity}\\
&=\frac{1}{n} \mathbb{E}\left(\mathds{1}(1\leq K)+\sum_{i=2}^n H(X_{A_i-A_{i-1}+1}|X_1)\mathds{1}(i\leq K)\right)\label{eq:stationarity}\\
&=\frac{1}{n} \left(\Pr(K\geq 1)+\sum_{i=2}^n \mathbb{E}\left(H(X_{G_i+1}|X_1)\mathds{1}(i\leq K)\right)\right),\label{eq:SumExpt}
\end{align}
where $\mathds{1}(T)$ is an indicator on the event $T$,~\eqref{eq:Markovity} follows since $\{X_n\}$ is a first-order Markov process, and~\eqref{eq:stationarity} follows from the stationarity of $\{X_n\}$. For any $2\leq i\leq n$ we have
\begin{align}
&\mathbb{E}\left(H(X_{G_i+1}|X_1)\mathds{1}(i\leq K)\right)\nonumber\\
&=\mathbb{E}_{G_i}\left(H(X_{G_i+1}|X_1)\Pr(K\geq i|G_i)\right)\nonumber\\
&=\mathbb{E}_{G_i}\left(H(X_{G_i+1}|X_1)\Pr(\Bin(n-G_i,\lambda)\geq i-1)\right)
%\label{eq:condexpt}
\nonumber
\end{align}
By the law of large numbers, for any $\epsilon>0$ and fixed $g_i$, there exists some $N_0$ such that for all $n>N_0$ holds
\begin{align}
\Pr(\Bin(n-g_i,\lambda)\geq i-1)\in\begin{cases}
[1-\epsilon,1) & i\leq (\lambda-\epsilon) n\\
(0,\epsilon] & i\geq (\lambda+\epsilon)n
\end{cases}
\nonumber
%\label{eq:Binth}
\end{align}
Combining this with~\eqref{eq:SumExpt} gives that
\begin{align}
\lim_{n\to\infty}\frac{1}{n} H(\bX_S|S)&=\lambda\mathbb{E}H(X_{G+1}|X_1)\nonumber\\
&=\lambda H(X_{G+1}|X_1,G),\label{eq:Markovlimit}
\end{align}
and our claim follows since $H(X_{G+1}|X_1,G)=\mathbb{E}h\left(q^{*G}\right)$ for stationary symmetric Markov processes.
\end{proof}

\begin{remark}
An expression similar to~\eqref{eq:Markovlimit} can also be recovered from~\cite[Corollary II.2]{lg14}.
\end{remark}

Our main result now follows directly from combining Propositions~\ref{prop:smgl2} and~\ref{prop:MarkovBEC} and using the continuity of the MGL function $\varphi(t)$.
\begin{theorem}
The entropy rate of the process obtained by passing a symmetric binary Markov process with transition probability $q$ through a BSC with crossover probability $\alpha$ satisfies
\begin{align}
\bar{H}(Y)\geq h\left(\alpha*h^{-1}\left(\mathbb{E}h\left(q^{*G}\right)\right)\right)\label{eq:hmmlb},
\end{align}
where $G$ is a geometric random variable with parameter $\lambda=(1-2\alpha)^2$.
\label{thm:hmmlb}
\end{theorem}

\section{Asymptotic Analysis and Numerical Examples}\label{sec:asymptotics}

In this section we evaluate the bound from Theorem~\ref{thm:hmmlb} in the limits of  $\alpha\to\tfrac{1}{2}$ and $q$ fixed (very noisy regime), and in the limit of $q\to\tfrac{1}{2}$ and $\alpha$ fixed (fast transitions regime).

\begin{theorem}
Let $q$ be fixed and $\alpha=\tfrac{1}{2}-\epsilon$. Then
\begin{align}
\bar{H}(Y)\geq 1-16\epsilon^4 \sum_{k=1}^\infty \frac{\log(e)}{2k(2k-1)}\frac{(1-2q)^{2k}}{1-(1-2q)^{2k}}+o(\epsilon^4).\nonumber
\end{align}
\label{thm:verynoisy}
\end{theorem}

\begin{proof}
By, e.g.,~\cite[Lemma 1]{os15it}, we have that for any $0\leq\alpha,\gamma\leq 1$ holds%Is there an approximation with o() instead of bound
\begin{align}
h(\alpha*\gamma)&\geq h(\gamma)+(1-h(\gamma))\cdot 4\alpha(1-\alpha)\nonumber\\
&=1-(1-2\alpha)^2\left(1-h(\gamma)\right).\label{eq:osbound}
\end{align}
Setting $\beta\triangleq \mathbb{E}h\left(q^{*G}\right)$, the RHS of~\eqref{eq:hmmlb} reads $h(\alpha*h^{-1}(\beta))$, which, by~\eqref{eq:osbound} and the parametrization $\alpha=\tfrac{1}{2}-\epsilon$, can be bounded as
\begin{align}
h\left(\alpha*h^{-1}(\beta)\right)&\geq 1-4\epsilon^2(1-\beta).\nonumber
\end{align}
It therefore, remains to approximate $\beta$. Recall the Taylor expansion of the binary entropy function
\begin{align}
h\left(\frac{1}{2}-p\right)=1-\sum_{k=1}^{\infty}\frac{\log(e)}{2k(2k-1)}(2p)^{2k}.\label{eq:EntTay}
\end{align}
%It follows that
%\begin{align}
%h\left(\frac{1}{2}-\delta\right)=1-2\log(e)\delta^2+o(\delta^4),
%\end{align}
%and
%\begin{align}
%h^{-1}\left(1-\delta\right)=\frac{1}{2}-\sqrt{\frac{\delta}{2\log(e)}}+o(\delta^3).
%\end{align}
%
%It remains to approximate $\beta$.
Using~\eqref{eq:qstar}, we have
\begin{align}
\beta&=\mathbb{E}h\left(\frac{1-(1-2q)^G}{2} \right)\nonumber\\
&=1-\sum_{k=1}^{\infty}\frac{\log(e)}{2k(2k-1)}\mathbb{E}(1-2q)^{2kG}\label{eq:infsum}\\
&=1-\sum_{k=1}^{\infty}\frac{\log(e)}{2k(2k-1)}\frac{\lambda(1-2q)^{2k}}{1-(1-\lambda)(1-2q)^{2k}},\label{eq:GeoExpt}
\end{align}
where~\eqref{eq:infsum} is justified since the sum $\sum_{k=1}^{\infty}\frac{\log(e)}{2k(2k-1)}\mathbb{E}(1-2q)^{2kG}$ converges, and in~\eqref{eq:GeoExpt} we have used the fact that $\mathbb{E} t^G=\tfrac{\lambda t}{1-(1-\lambda)t}$. To further approximate~\eqref{eq:GeoExpt}, we write
\begin{align}
\frac{\lambda(1-2q)^{2k}}{1-(1-\lambda)(1-2q)^{2k}}&=\frac{\lambda(1-2q)^{2k}}{1-(1-2q)^{2k}}\cdot\frac{1}{1+\frac{\lambda(1-2q)^{2k}}{1-(1-2q)^{2k}}}\nonumber\\
&=\frac{\lambda(1-2q)^{2k}}{1-(1-2q)^{2k}}\cdot\left(1+\m{O}(\lambda)\right)\nonumber\\
&=4\epsilon^2\frac{(1-2q)^{2k}}{1-(1-2q)^{2k}}+\m{O}\left(\epsilon^4\right).\nonumber
\end{align}
Consequently,
\begin{align}
\beta&=1-4\epsilon^2\sum_{k=1}^{\infty}\frac{\log(e)}{2k(k-1)}\frac{(1-2q)^{2k}}{1-(1-2q)^{2k}}+\m{O}\left(\epsilon^4\right).\nonumber
\end{align}
which yields the desired result.
\end{proof}

To date, the best known upper and lower bounds on $\bar{H}(Y)$ in the very noisy regime were the ones found in~\cite[Theorem 4.13]{ow11}. In particular, the ratio $\tfrac{1-\bar{H}(Y)}{\epsilon^4}$ was bounded form above and below.
The upper and lower bounds from~\cite[Theorem 4.13]{ow11} on $(1-\bar{H}(Y))/\epsilon^4$ are plotted in Figure~\ref{fig:verynoisy}, along with the upper bound from Theorem~\ref{thm:verynoisy}. It is seen that Theorem~\ref{thm:verynoisy} improves upon the best known lower bounds on $\bar{H}(Y)$ in the limit of $\alpha\to\tfrac{1}{2}$. Furthermore, unlike~\cite[Theorem 4.13]{ow11} that only holds for $q\geq\tfrac{1}{4}$, our result holds for all $q$.

%which state that for $\alpha=\tfrac{1}{2}-\epsilon$ and $\tfrac{1}{4}\leq q\leq \tfrac{1}{2}$ hold
%\begin{align}
%&\frac{2\log(e)(1-2q)^2(1-12q+48q^2-64q^3+32q^4)}{q^2}\nonumber\\
%&\leq\liminf_{\epsilon\to 0}\frac{1-\bar{H}(Y)}{\epsilon^4}
%\leq\limsup_{\epsilon\to 0}\frac{1-\bar{H}(Y)}{\epsilon^4}\nonumber\\
%&\leq \frac{2\log(e)(1-2q)^2(1-4q+16q^2-32q^3+32q^4)}{q^2}
%\label{eq:ownoisy}
%\end{align}
%Theorem~\ref{thm:verynoisy} states that
%\begin{align}
%\limsup_{\epsilon\to 0}\frac{1-\bar{H}(Y)}{\epsilon^4}\leq 16 \sum_{k=1}^\infty \frac{\log(e)}{2k(2k-1)}\frac{(1-2q)^{2k}}{1-(1-2q)^{2k}}.
%\label{eq:noisy_rearrangd}
%\end{align}

\begin{figure}[h]
\begin{center}
\includegraphics[width=1 \columnwidth]{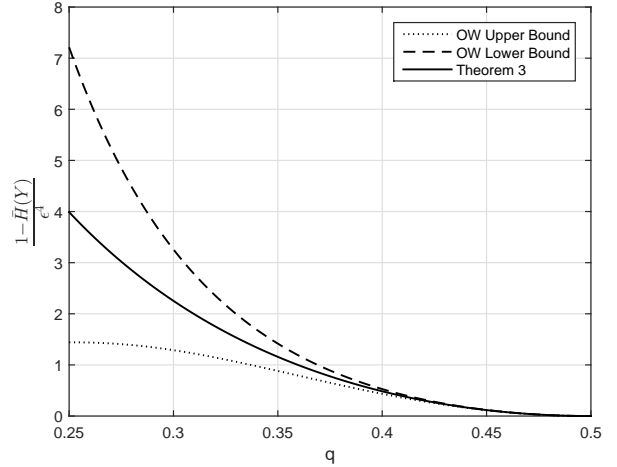}
\end{center}
\caption{Comparison between the bounds from~\cite[Theorem 4.13]{ow11} on $\bar{H}(Y)$ in the very noisy regime, and the new bound from Theorem~\ref{thm:verynoisy}.}
\label{fig:verynoisy}
\end{figure}

\begin{figure*}[!ht]
\begin{center}
\subfloat[$\alpha=0.11$]{
\includegraphics[width=0.45\textwidth]{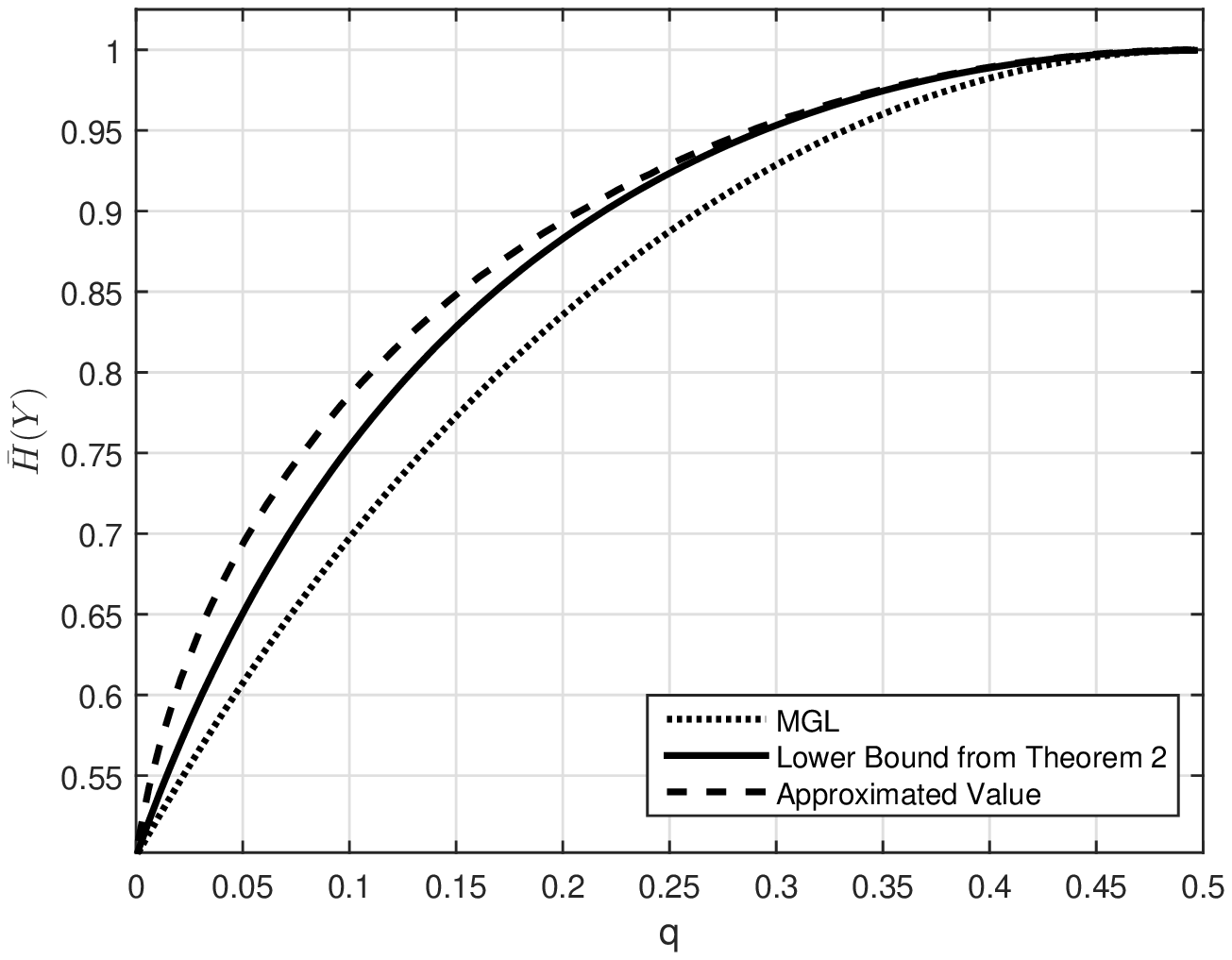}
\label{fig:comp1}}
\qquad
\subfloat[$q=0.11$]{
\includegraphics[width=0.45\textwidth]{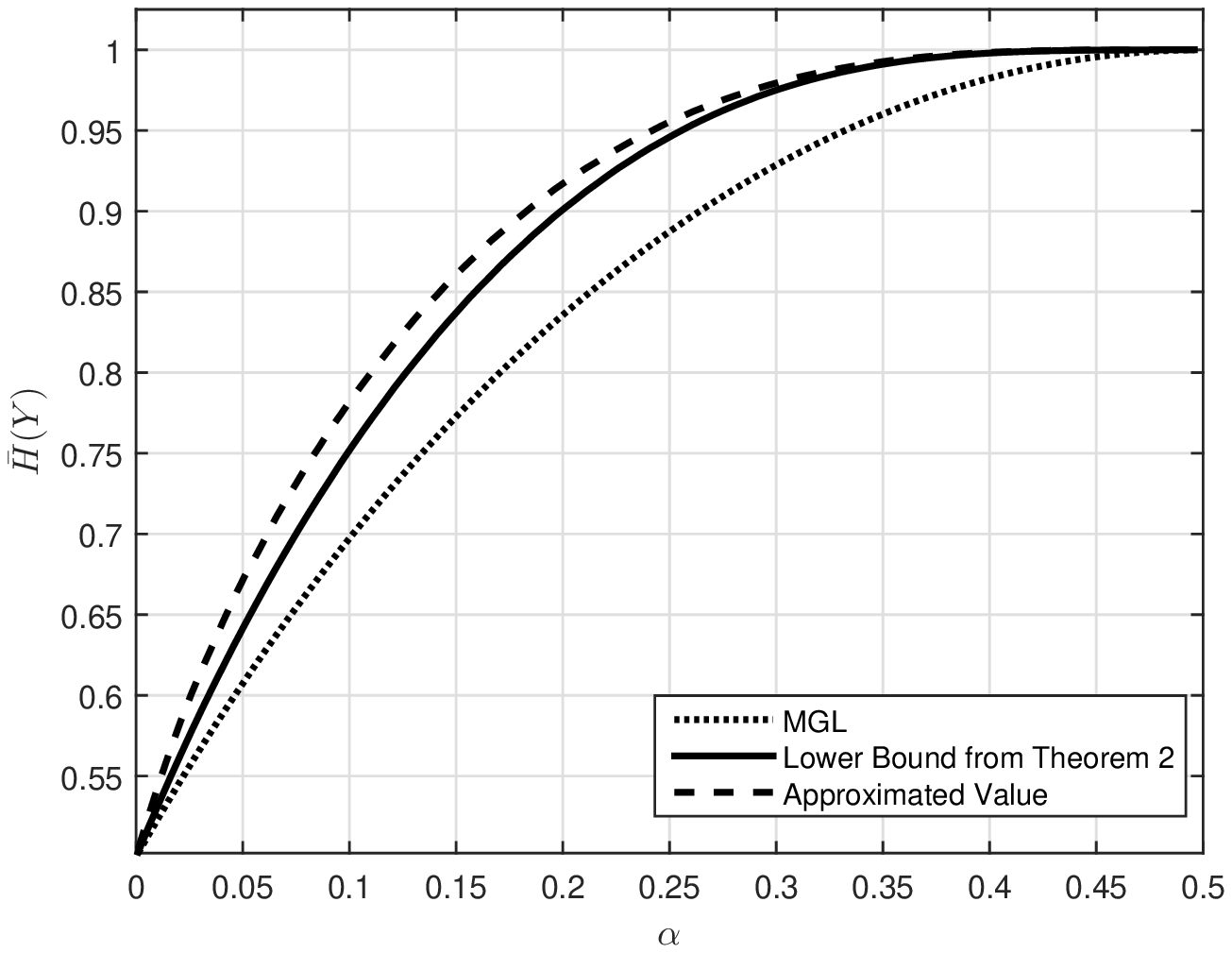}
\label{fig:comp2}}
\end{center}
\caption{Comparison between the lower bound~\eqref{eq:hmmlb}, MGL based bound~\eqref{eq:mglmarkov}, and the approximate value of $\bar{H}(Y)$ computed using~\cite[Algorithm 4.25]{ow11}.}
\end{figure*}

Next, we move to show that the lower bound from Theorem~\ref{thm:hmmlb} is tight in the extreme regime of fast transitions, i.e., $q\to\tfrac{1}{2}$ and $\alpha$ fixed. Let $q=\tfrac{1}{2}-\epsilon$. With this parametrization,~\eqref{eq:GeoExpt} reads
\begin{align}
\beta&=1-\sum_{k=1}^{\infty}\frac{\log(e)}{2k(2k-1)}\frac{\lambda(2\epsilon)^{2k}}{1-(1-\lambda)(2\epsilon)^{2k}}\\
&=1-2\log(e)\lambda\epsilon^2+\m{O}(\epsilon^4).\nonumber
\end{align}
Now, using~\eqref{eq:osbound} and Theorem~\ref{thm:hmmlb}, we have the following proposition.
\begin{proposition}
Let $\alpha$ be fixed and $q=\tfrac{1}{2}-\epsilon$. Then
\begin{align}
1-\bar{H}(Y)&\leq \lambda(1-\beta)\nonumber\\
&=2\log(e)\lambda^2\epsilon^2+\m{O}(\epsilon^4)\nonumber\\
&=2\log(e)(1-2\alpha)^4\epsilon^2+\m{O}(\epsilon^4)\nonumber
\end{align}
\end{proposition}
In~\cite[Theorem 4.12]{ow11} it was proved that for $0\leq \alpha\leq \tfrac{1}{2}$, $q=\tfrac{1}{2}-\epsilon$, as $\epsilon\downarrow0$
\begin{align}
1-\bar{H}(Y)=2\log(e)(1-2\alpha)^4\epsilon^2+o(\epsilon^3).\nonumber
\end{align}
It therefore follows that the bound from Theorem~\ref{thm:hmmlb} becomes tight as $q\to\tfrac{1}{2}$.% and the its rate of convergence to the asymptotic value is even slightly faster than the best known previous results.

It can be shown that in the regimes of very high-SNR ($\alpha\to 0$ and $q$ fixed) and rare transitions ($q\to 0$ and $\alpha$ fixed), the bound from Theorem~\ref{thm:hmmlb} is looser than the bounds found in~\cite[Theorem 4.11]{ow11} and in~\cite{pq11}, respectively.

For any pair of finite values of $(\alpha,q)$, the entropy rate $\bar{H}(Y)$ can be approximated to an arbitrary precision. For example,~\cite[Theorem 4.5.1]{coverthomas} shows that
\begin{align}
H(Y_n|Y_{n-1}\ldots,Y_1,X_1)\leq \bar{H}(Y)\leq H(Y_n|Y_{n-1}\ldots,Y_1)\nonumber
\end{align}
and the two bounds converge to the same limit as $n\to\infty$. Unfortunately, the computational complexity of the lower bound (as well as the upper bound) above, is exponential in $n$.\footnote{Although, as shown by Birch~\cite{birch62}, the gap between the two bounds also decreases exponentially (but possibly with a small exponent) in $n$.} To that end, various works introduced different algorithms for approximating $\bar{H}(Y)$~\cite{pfisterthesis,jss04,lg09,ow11}, each algorithm exhibiting a different trade-off between approximation accuracy and complexity.

To obtain a better appreciation of the tightness of the bound from Theorem~\ref{thm:hmmlb} for finite values of $(\alpha,q)$, we numerically compare it to the output of one such approximation algorithm. In particular, we use~\cite[Algorithm 4.25]{ow11} to approximate $\bar{H}(Y)$, where the algorithm parameters are chosen to ensure high enough accuracy, and plot the results alongside with the lower bound of Theorem~\ref{thm:hmmlb}. We also plot the lower bound~\eqref{eq:mglmarkov} obtained by simply applying Mrs. Gerber's Lemma. The results for fixed $\alpha=0.11$ and varying $q$ are shown in Figures~\ref{fig:comp1}, and those for fixed $q=0.11$ and varying $\alpha$, in Figure~\ref{fig:comp2}.

\section{Nonsymmetric Markov Chains}
\label{sec:nonsym}

In this section we extend our lower bound from Theorem~\ref{thm:hmmlb} to the case where the input to the BSC is a nonsymmetric Markov process. Let
\begin{align}
\bP=\left[
      \begin{array}{cc}
        1-q_{01} & q_{01}  \\
        q_{10} & 1-q_{10} \\
      \end{array}
    \right]\nonumber
\end{align}
be a transition probability matrix, and $\bpi=[\pi_0 \ \pi_1]$ be a stationary distribution for $\bP$, such that $\bpi\bP=\bpi$. Let $\{X_n\}$ be a stationary first-order Markov process with transition probability matrix $\bP$, such that $X_1\sim\Ber(\pi_1)$ and for $n=2,3,\ldots$ holds $\Pr(X_n=j|X_{n-1}=i)=\bP_{ij}$. For $k=1,2,\ldots$, we define the quantities
\begin{align}
q^{\#k}_{ij}\triangleq\left(\bP^k\right)_{ij}=\Pr\left(X_n=j|X_{n-k}=i\right).\nonumber
\end{align}

The following is an extension of Proposition~\ref{prop:MarkovBEC} for nonsymmetric hidden Markov processes.

\begin{proposition}
Let $\{X_n\}$ be a stationary first-order Markov process with transition probability matrix $\bP$ and stationary distribution $\bpi$. Let $0< \lambda< 1$, and let $S$ be a random subset of $[n]$ generated by independently sampling each element
$i$ with probability $\lambda$. Then
\begin{align}
\lim_{n\to\infty}\frac{H(\bX_S|S)}{\lambda n}=\pi_0\mathbb{E}h\left(q_{01}^{\#G}\right)+\pi_1\mathbb{E}h\left(q_{10}^{\#G}\right),\label{eq:nonsyment}
\end{align}
where $G$ is a geometric random variable with parameter $\lambda$, i.e. $\Pr(G=g)=(1-\lambda)^{g-1}\lambda$ for $g=1,2,\ldots$.
\label{prop:nonsymMarkovBEC}
\end{proposition}

\begin{proof}
The proof is similar to that of Proposition~\ref{prop:MarkovBEC} up to equation~\eqref{eq:Markovlimit}, where now $H(X_{G+1}|X_1,G)=\pi_0\mathbb{E}h\left(q_{01}^{\#G}\right)+\pi_1\mathbb{E}h\left(q_{10}^{\#G}\right)$.
\end{proof}

Combining this with Proposition~\ref{prop:smgl2} and the continuity of the MGL function $\varphi(t)$ gives the following.
\begin{theorem}
The entropy rate of the process obtained by passing a stationary first-order Markov process with transition probability matrix $\bP$ and stationary distribution $\bpi$ through a BSC with crossover probability $\alpha$ satisfies
\begin{align}
\bar{H}(Y)\geq h\left(\alpha*h^{-1}\left(\pi_0\mathbb{E}h\left(q_{01}^{\#G}\right)+\pi_1\mathbb{E}h\left(q_{10}^{\#G}\right)\right)\right)\nonumber,
\end{align}
where $G$ is a geometric random variable with parameter $\lambda=(1-2\alpha)^2$.
\label{thm:hmmnonsymlb}
\end{theorem}

\subsection{Example: Processes Satisfying the $(1,\infty)$-RLL Constraint}
In this subsection we lower bound the entropy rate of a nonsymmetric first order Markov process, with $q_{01}=q$ and $q_{10}=1$, passed through a BSC with crossover probability $\alpha$. This underlying Markov process satisfies the so-called $(1,\infty)$-RLL constraint, where no consecutive ones are allowed to appear in a sequence. It is not difficult to verify that for this choice of $q_{01}$ and $q_{10}$ we have
\begin{align}
\pi_{0}&=\frac{1}{1+q}, \ \ \ \ \ \ \ \ \  \ \ \ \ \ \pi_{1}=\frac{q}{1+q},\nonumber\\
q^{\#k}_{01}&=\frac{q+(-q)^{k+1}}{1+q}, \ \ \ \ q^{\#k}_{10}=\frac{1-(-q)^{k}}{1+q}.\nonumber
\end{align}
In this case we have $\bar{H}(Y)\geq h(\alpha*h^{-1}(\beta))$, where
\begin{align}
\beta\triangleq\mathbb{E}\left(\frac{1}{1+q}h\left(\frac{1-(-q)^{G+1}}{1+q}\right)+\frac{q}{1+q}h\left(\frac{1-(-q)^{G}}{1+q}\right)\right).
\label{eq:betanonsym}
\end{align}
By the concavity of $h(\cdot)$, for any natural number $g$ hold
\begin{align}
&h\left(\frac{1-(-q)^{g+1}}{1+q}\right)=h\left(\frac{(1-q^g)+q^g(1-q(-1)^{g+1})}{1+q}\right)\nonumber\\
&\geq (1-q^g)h\left(\frac{1}{1+q}\right)+q^g h\left(\frac{1-q(-1)^{g+1}}{1+q}\right).\label{eq:hq01g}
\end{align}
and
\begin{align}
&h\left(\frac{1-(-q)^{g}}{1+q}\right)=h\left(\frac{(1-q^{g-1})+q^{g-1}(1-q(-1)^{g})}{1+q}\right)\nonumber\\
&\geq (1-q^{g-1})h\left(\frac{1}{1+q}\right)+q^{g-1} h\left(\frac{1-q(-1)^g}{1+q}\right).\label{eq:hq10g}
\end{align}
Substituting~\eqref{eq:hq01g} and~\eqref{eq:hq10g} into~\eqref{eq:betanonsym}, we obtain
\begin{align}
\beta\geq \left(1-2\frac{\mathbb{E}q^G}{1+q}\right)h\left(\frac{1}{1+q}\right)+\frac{\mathbb{E}q^G}{1+q}h\left(\frac{1-q}{1+q} \right).\label{eq:betanonsym2}
\end{align}
Now, using again the fact that $\mathbb{E}(q^G)=\tfrac{\lambda q}{1-(1-\lambda)q}$, and invoking Theorem~\ref{thm:hmmnonsymlb} gives the following result.
\begin{theorem}
The entropy rate of a nonsymmetric stationary binary first-order Markov process with transition probabilities $q_{01}=q$ and $q_{10}=1$, passed through a BSC with crossover probability $\alpha$, is lower bounded as $\bar{H}(Y)\geq h\left(\alpha*h^{-1}(\gamma)\right)$, where
\begin{align}
\gamma&\triangleq
h\left(\frac{1}{1+q}\right)\nonumber\\
&-\frac{(1-2\alpha)^2 q}{(1+q)(1-4\alpha(1-\alpha)q)}\left(2h\left(\frac{1}{1+q}\right)- h\left(\frac{1-q}{1+q} \right)\right).\nonumber
\end{align}
\label{thm:hmmRLL}
\end{theorem}

The following Corollary of Theorem~\ref{thm:hmmRLL}, shows that our bound becomes tight as $\alpha\to\tfrac{1}{2}$, and partially recovers the results of~\cite[Section 4.2]{Pfister11} and~\cite[Appendix E]{hm07}.

\begin{corollary}
For the very noisy regime, where $\alpha=\tfrac{1}{2}-\epsilon$ and $0\leq q< 1$, we have
\begin{align}
\bar{H}(Y)= 1-2\log(e)\left(\frac{1-q}{1+q}\right)^2\epsilon^2+\m{O}(\epsilon^4).\nonumber
\end{align}
\end{corollary}
\begin{proof}
Clearly, $\bar{H}(Y)\leq H(Y_n)=h(\alpha*\pi_1)$. From~\cite[equation (11)]{os15it} combined with~\eqref{eq:EntTay}  we have
\begin{align}
h(\alpha*\pi_1)&=1-\sum_{k=1}^{\infty}\frac{\log(e)}{2k(2k-1)}\left(2\epsilon(1-2\pi_1)\right)^{2k}\label{eq:convTay}\\
&=1-2\log(e)(1-2\pi_1)^2\epsilon^2+\m{O}(\epsilon^4),\nonumber
\end{align}
which establishes our upper bound. For the lower bound, note that $\gamma\geq h(\pi_1)-c\epsilon^2$ for some universal constant $c>0$. From the concavity of $h(\cdot)$ we have that for all $0\leq x\leq 1$ holds $h(x)\leq h(\pi_1)+h'(\pi_1)(x-\pi_1)$. Thus, using the monotonicity of $h^{-1}(\cdot)$ and the fact that $h'(\pi_1)>0$ for all $0\leq q<1$, we have
\begin{align}
h^{-1}(\gamma)\geq h^{-1}(h(\pi_1)-c\epsilon^2)\geq \pi_1-\frac{c}{h'(\pi_1)}\epsilon^2.\label{eq:invgamma}
\end{align}
Now, by Theorem~\ref{thm:hmmRLL},~\eqref{eq:invgamma} and~\eqref{eq:convTay}, we have
\begin{align}
\bar{H}(Y)&\geq h(\alpha*h^{-1}(\gamma))\geq h\left(\alpha*\left(\pi_{1}-\frac{c\epsilon^2}{h'(\pi_1)}\right)\right)\nonumber\\
&=1-2\log(e)(1-2\pi_1)^2\epsilon^2+\m{O}(\epsilon^4).\nonumber
\end{align}
\end{proof}

\section*{Acknowledgment}
The author is grateful to Alex Samorodnitsky for many valuable discussions and observations, and to Yury Polyanskiy and Yihong Wu for sharing an early draft of~\cite{pw16}.

\bibliographystyle{IEEEtran}
\bibliography{Dissertation_bib}

\end{document}